\documentclass[journal]{IEEEtran}
\ifCLASSINFOpdf
\else
\fi
\usepackage{graphicx}
\usepackage{amsmath}
\usepackage{amssymb}
\usepackage{cite}
\usepackage{subfigure}
\usepackage{booktabs}
\usepackage{color}
\usepackage{balance}
\usepackage{amsthm}
\usepackage{subfig}
\usepackage{subfloat}

\newtheorem{corollary}{Corollary}

\begin{document}
%
\title{Joint Block Low Rank and Sparse Matrix Recovery in Array Self-Calibration Off-Grid DoA Estimation}
%
%
%

\author{Cheng-Yu~Hung,~\IEEEmembership{Member,~IEEE,}
        and~Mostafa~Kaveh,~\IEEEmembership{Life~Fellow,~IEEE}
\thanks{C. Y. Hung was with the Department
of Electrical and Computer Engineering, University of Minnesota - Twin Cities, Minneapolis,
MN, 55455 USA e-mail: hungx086@umn.edu.}
\thanks{M. Kaveh is with University of Minnesota. E-mail: mos@umn.edu.}
}

\maketitle

\begin{abstract}
This letter addresses the estimation of directions-of-arrival (DoA) by a sensor array using a sparse model in the presence of array calibration errors and off-grid directions.  The received signal utilizes previously used models for unknown errors in calibration and structured linear representation of the off-grid effect.  A convex optimization problem is formulated with an objective function to promote two-layer joint block-sparsity with its second-order cone programming (SOCP) representation. The performance of the proposed method is demonstrated by numerical simulations and compared with the Cramer-Rao Bound (CRB), and several previously proposed methods.
\end{abstract}


%
\IEEEpeerreviewmaketitle

\section{Introduction}
%
%
%
%

\IEEEPARstart{A}{rray} signal processing, in general, and estimation of the Directions-of-Arrivals (DoA), in particular, require the spatial signatures of incident waves from angles of interest. The spatial signatures are usually obtained through array calibration. However, maintaining calibration may be difficult due to array element gain and/or phase changes caused by variations in environmental conditions or relative element positions. Self-calibration has been suggested as a way to mitigate perturbations to the array signature model \cite{astely1999spatial}.

Examples of array processing degradations due to unknown calibration errors can be found in  
\cite{wax1996performance, yang1995effects}. In \cite{fuhrmann1994estimation}, sensor gain and phase errors are estimated based on the assumption of perfect knowledge of the DoAs. In \cite{soon1994subspace}, an alternating approach is used to estimate unknown
gains and phases, and the DoAs. An eigenstructure-based (ES) method is developed in \cite{weiss1990eigenstructure} to estimate the calibration errors and the DoAs. In \cite{ling2015self}, the SparseLift method is proposed to solve a biconvex compressed sensing problem for the joint estimation of calibration errors and DoAs in the single measurement vector (SMV) model when the unknown directions belong to the set of angles in the search grid, the so-called on-grid model. In \cite{hung2017Low}, the model of \cite{ling2015self} is extended to the situation of multiple measurement snapshots, or the multiple measurement vector
(MMV) model, which naturally results in improvement of the accuracy of DoA estimates. This is accomplished by solving a modified nuclear norm minimization problem together with singular value decomposition (SVD) to reduce computational complexity.

When the compressed sensing or the sparsity frameworks are used, the quality of the sensing or search model matrices can also significantly impact the accuracy of the resulting estimators. For example, the effects of basis mismatch in compressed sensing is investigated in \cite{chi2011sensitivity}. Performance degradation of structured perturbations on DoA estimation for sparse models, which is called the off-grid effect, is studied in \cite{yang2012robustly,zheng2013sparse,tan2014joint,hung2014directions}. In \cite{tan2014joint,hung2017smoothed}, iterative algorithms are developed for off-grid DoA estimation. The off-grid effect on DoA estimation can be avoided by using the super-resolution framework in a continuous-domain manner \cite{tan2014direction,hung2016super}, but the DoA resolution performance is limited by the number of array elements \cite{candes2014towards,tang2012compressive}.


In this letter, the work in \cite{hung2017Low} is extended from the on-grid array self-calibration model to the more practical off-grid one. In contrast to \cite{hung2017Low} and \cite{ling2015self}, the two uncertainties
mentioned above, i.e. unknown array gain and phase responses and the off-grid effect are jointly modeled for the received signals for the general case of multiple measurement vectors (MMV). 
Using the perturbation structure for off-grid DoAs \cite{zhu2011sparsity} and multiple
measured snapshots in the self-calibration model, a new objective function is proposed to formulate a convex optimization problem. We give the second-order cone programing (SOCP) \cite{boyd2009convex,ben2001lectures} representation such that optimal solution can be obtained by using the interior point method. The performance of the proposed method is demonstrated by
numerical simulations and compared with the Cramer-Rao Bound (CRB) \cite{friedlander2014array}, the ES method \cite{weiss1990eigenstructure}, Ling’s method \cite{ling2015self}, and the MMV-SC method \cite{hung2017Low}.

\section{Signal Model for Self-Calibration}

\subsection{MMV Model for Self-Calibration with DoA Estimation}
Consider the DoA estimation problem with a uniform linear array (ULA) of $M$ sensors, and $L$ snapshots. Suppose there are $K$ far-field narrowband plane waves impinging on the array from angles $\theta_1, \dots, \theta_K$. The self-calibration MMV model \cite{hung2017Low} is expressed as  
\begin{align} \label{multi-snap}
{\bf Y} ={\bf DA}{{\bf S}} +{\bf N},  ~~~~{\bf D}=\text{diag}(\bf Bh) 
\end{align}
where ${\bf Y}=[{\bf y}_1,\cdots,{\bf y}_L]\in \mathbb{C}^{M\times L}$ is the observation matrix. The measurement matrix ${\bf A}=[{\bf a}(\theta_1),\cdots,{\bf a}(\theta_K)]\in \mathbb{C}^{M\times K}$ is composed of the steering vectors $\{ {\bf a}(\theta_i)=[e^{-j(-(M-1)/2)2\pi \frac{d}{\lambda}sin\theta_i},\dots,e^{-j((M-1)/2)2\pi \frac{d}{\lambda}sin\theta_i}]^T \}_{i=1}^K$ with wavelength $\lambda$, $d$ is the distance between sensors, ${\bf S}=[{\bf s}_1,\cdots,{\bf s}_L]\in \mathbb{C}^{K\times L}$ (${\bf s}_i \in\mathbb{C}^{K\times 1},\forall i$  represents the arriving stochastic signal vector with zero-mean, and covariance matrix $ {\bf C}_s$.) is the signal matrix of interest, and ${\bf N}\in \mathbb{C}^{M\times L}$ is an additive white Gaussian noise matrix whose elements are zero-mean and $\sigma_n^2$-variance. 

 
Matrix ${\bf D} \in \mathbb{C}^{M\times M}$ is parameterized by an unknown parameter vector ${\bf h} \in \mathbb{C}^{m\times 1}$, which captures the unknown calibration of the sensors. The calibration case of interest is when ${\bf D}({\bf h})= \text{diag}({\bf B h})$ is a diagonal matrix in which its diagonal entries represent unknown complex gains for each antenna. ${\bf B} \in \mathbb{C}^{M\times m} (m<M)$ is assumed to be a known matrix, which is used to model the situation when the diagonal elements of ${\bf D}$ change slowly entry-wise \cite{ling2015self}.

\section{Array Self-Calibration with Off-Grid Directions}
We discretized the angle space into a grid of directions,  which are denoted by $\{ \phi_1, \phi_2,\cdots,\phi_N \}$ where $N$ is the number of discrete directions and $N \gg K$. Suppose the actual DoAs $\{ \theta_1, \theta_2, \cdots, \theta_K\}$ belong to the grid of interest represented by $\{ \phi_1, \phi_2,\cdots,\phi_N \}$. Then, Equation (\ref{multi-snap}) can be transformed into a sparse model as:
\begin{align} 
{\bf Y} ={{\bf D}\bar{\bf A}}{\bar{\bf S}} +{\bf N},  ~~~~{\bf D}=\text{diag}(\bf Bh) ,
\end{align}
where $\bar{\bf A}=[{\bf a}(\phi_1),\cdots,{\bf a}(\phi_N)]\in \mathbb{C}^{M\times N}$, and  $\bar{\bf S}=[\bar{\bf s}_1,\cdots,\bar{\bf s}_L]\in \mathbb{C}^{N\times L}$ is a sparse matrix with each column $ \bar{\bf s}_i \in\mathbb{C}^{N\times 1}$ sparse.
However, in reality and with a high probability, the actual DoAs will not belong to the grid so that off-grid errors occur in the model. Thus, using the first order model approximation to account for off-grid directions \cite{zhu2011sparsity}, we have:
\begin{align} \label{sc_offg_mmv_p}
 {\bf Y} &\cong {\bf D}( \bar{\bf A}+\bar{{\bf B}}\Gamma){\bar{\bf S}} +{\bf N} \\ \nonumber
&= {\bf D}(\bar{\bf A}{\bar{\bf S}}+\bar{\bf B}{\bf P}) +{\bf N} \\\nonumber
 &= {\bf D}[\bar{\bf A}, \bar{\bf B}]{\bf X} +{\bf N} \\\nonumber
  &= {\bf D}{\bf G}{\bf X} +{\bf N} , 
\end{align}
where $\bar{{\bf B}}=[\frac{\partial{\bf a}(\phi_1)}{\partial \phi_1},\dots,\frac{\partial{\bf a}(\phi_N)}{\partial \phi_N}]\in {\mathbb C^{M\times N}}$, ${\boldsymbol{\beta}}=[\beta_1,\dots,\beta_N]^T$, $\Gamma=diag(\boldsymbol{\beta})$, ${\bf P}=\Gamma{\bar{\bf S}}$,  ${\bf G}=[\bar{\bf A}, \bar{\bf B}]$, and ${\bf X}=[\bar{\bf S}^T,{\bf P}^T]^T\in {\mathbb C^{2N\times L}}$ with each column a sparse vector ${\bf x}_i=[\bar{\bf s}_i^T,{\bf p}_i^T]^T\in {\mathbb C^{2N\times 1}}$, and ${\bf p}_i={\boldsymbol \beta}_i\odot \bar{\bf s}_i$ where $\odot$ denotes the Hadamard product.  
 It is noted that $N \gg M>m $, and each column ${\bf x}_i$ is $2K$-jointly sparse, which means that $\bar{\bf s}_i^T$ and ${\bf p}_i^T$ have the same non-zero locations.


\subsection{The Proposed Method}


In \cite{hung2017Low}, the joint low rank and sparse matrix recovery is proposed for the on-grid array calibration model.
By considering the off-grid array calibration model (\ref{sc_offg_mmv_p}), one can follow the same approach of \cite{hung2017Low}, supposing noiseless condition to define
\begin{itemize}
\item $
 {\bf Y}_{i,:}=[{ y}_{i,1} , \cdots , { y}_{i,L} ]\\
={\bf b}_i^H [{\bf h}{\bf x}_1^T, \cdots , {\bf h}{\bf x}_L^T]\begin{bmatrix}
      {\bf g}_i & {\bf 0}  & {\bf 0}\\
      {\bf 0} & \ddots & {\bf 0}\\
      {\bf 0} & {\bf 0} & {\bf g}_i 
    \end{bmatrix}={\bf b}_i^H {\tilde{\textbf{\textit X}}} {\tilde{\bf G}_i}  ,
$\\
where ${\bf b}_i$ is the $i$-th column of ${\bf B}^H$, and ${\bf g}^T_i$ is the $i$-th row of ${\bf G}$.
\item $
 {\tilde{\textbf{\textit X}}} \overset{\Delta}{=} {\bf h}[{\bf x}_1^T, \cdots, {\bf x}_L^T]={\bf h}\begin{bmatrix}
      {\bf x}_1 \\
      \vdots  \\
      {\bf x}_L
    \end{bmatrix}^T    \in {\mathbb C}^{m \times 2LN} \\
{\tilde{\bf G}_i}  =\begin{bmatrix}
      {\bf g}_i & {\bf 0}  & {\bf 0}\\
      {\bf 0} & \ddots & {\bf 0}\\
      {\bf 0} & {\bf 0} & {\bf g}_i 
    \end{bmatrix} \in {\mathbb C}^{2LN \times L}
$
\item 
 Linear operator ${\mathcal A}: \mathbb{C}^{m\times 2LN}\rightarrow\mathbb{C}^{M\times L}$ s.t.
$
  {\mathcal A}(\tilde{{\textbf{\textit X}}})   \overset{\Delta}{=}  \{{\bf b}_i^H\tilde{{\textbf{\textit X}}}\tilde{{\bf G}_i} \}_{i=1}^M
$
\item 
Matrix representation $\Phi : ML \times 2mLN$ of $\mathcal A$ such that
\begin{align} \label{remark1}
& \Phi \text{vec}(\tilde{\textbf{\textit X}}) = \text{vec}({\mathcal A}(\tilde{\textbf{\textit X}}))=\text{vec}({\bf Y}^T),\\ \nonumber
 & \Phi  = [ \varphi_1,\cdots,\varphi_i,\cdots,\varphi_M]^H \in {\mathbb C}^{ML \times 2mLN}, \\ \nonumber
& \varphi_i =\tilde{\bf G}_i^* \otimes {\bf b}_i  \in {\mathbb C}^{2mLN\times L}  \nonumber.
\end{align}

\end{itemize}
 Thus, in terms of (\ref{remark1}) and using the property of group sparsity, a  convex optimization problem is formulated as
\begin{align}\label{formulation_1}
&\arg\min_{\tilde{\textbf{\textit X}}} ~  ||\tilde{\textbf{\textit X}}||_* + \lambda ||\tilde{\textbf{\textit X}}||_{2,1}  \\ \nonumber
&\text{subject to }|| \Phi \text{vec}(\tilde{\textbf{\textit X}})-\text{vec}({\bf Y}^T)||_2\leq \eta ,
\end{align}
where the nuclear norm $||\tilde{\textbf{\textit X}}||_*$ is the sum of singular values of matrix $\tilde{\textbf{\textit X}}$, $||\tilde{\textbf{\textit X}}||_{2,1} =\sum_{i=1}^{2LN} \| \tilde{\textbf{\textit X}}_{:,i}\|_2$, $\tilde{\textbf{\textit X}}_{:,i}$ denotes the $i$-th column of $\tilde{\textbf{\textit X}}$, and $\lambda > 0$. 
We note that the column size of $\tilde{\textbf{\textit X}}$ in (\ref{formulation_1}) is double that of $\tilde{\textbf{\textit X}}$ in (16) of \cite{hung2017Low} due to modeling the off-grid errors.
\begin{figure}[tb]
\begin{center}
\includegraphics[width=\columnwidth]{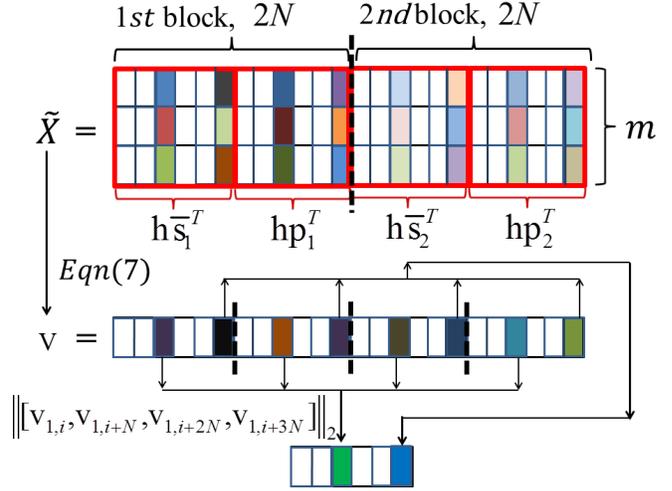}
\caption{Illustration of joint block-sparsity of matrix ${\tilde{\textbf{\textit X}}}$ with $K=2$, $m=3$, $N=6$, and $L=2$ snapshots (so there are two blocks). 
}
\label{fig:block_sparsity}
\end{center}
\end{figure}
\\
Furthermore, each row of $\tilde{\textbf{\textit X}}$  has joint block-sparsity patterns as shown in Figure \ref{fig:block_sparsity}. And since $0\leq |\beta_i|\leq r$ and $r=\frac{|{\phi_i-\phi_{i+1}}|}{2}$ is half the size of the grid interval, we can define a new norm for $\tilde{\textbf{\textit X}}$ to take advantage of the joint block-sparsity property as
\begin{align}
\|\tilde{\textbf{\textit X}}\|_{2,1,2} \overset{\Delta}{=} \| {\bf v} \|_{1,2} , 
\end{align}
where
\begin{align}
{\bf v} \overset{\Delta}{=}  [ \| \tilde{\textbf{\textit X}}_{:,1}  \|_2,  \| \tilde{\textbf{\textit X}}_{:,2}  \|_2 , \cdots,  \| \tilde{\textbf{\textit X}}_{:,2LN}  \|_2] \in {\mathbb R}^{1 \times 2LN},\\
\| {\bf v} \|_{1,2}  \overset{\Delta}{=}  \sum_{i=1}^{N} \| [{\bf v}_{1,i}, {\bf v}_{1,i+N}, \cdots, {\bf v}_{1,i+(2L-1)N}  ] \|_2. 
\end{align}
Note that ${\bf v}$ is a jointly sparse vector, and ${\bf v}_{1,i} =  \| \tilde{\textbf{\textit X}}_{:,i}  \|_2$.
So, we can solve a new convex optimization problem as follows
\begin{align}\label{2_1_2_norm}
&\arg\min_{\tilde{\textbf{\textit X}},{\bf v}} ~~~   ||\tilde{\textbf{\textit X}}||_* + \lambda \|\tilde{\textbf{\textit X}}\|_{2,1,2} \\ \label{ori_ineq_1}
&\text{subject to } || \Phi \text{vec}(\tilde{\textbf{\textit X}})-\text{vec}({\bf Y}^T)||_2\leq \eta  \\ \label{ori_ineq_2}
&~~~~~~~~~~~~~~{\bf v} = [ \| \tilde{\textbf{\textit X}}_{:,1}  \|_2,  \| \tilde{\textbf{\textit X}}_{:,2}  \|_2 , \cdots,  \| \tilde{\textbf{\textit X}}_{:,2LN}  \|_2]  \\ \label{ori_ineq_3}
&~~~~~~~~~~~~~~{\bf v}\geq 0   \\ \label{ori_ineq_4}
&~~~~~~~~~~~~~~{\bf v}_{1,l N+1:(l+1)N} \leq r{\bf v}_{1,(l-1)N+1:l N},  \\   \nonumber  
&~~~~~~~~~~~~~~~~~~~~~~~~~~~~~~~~~~~~\forall l = 1, 3, \cdots, (2L-1)     \nonumber
\end{align}
The last two new constraints (\ref{ori_ineq_3}) and (\ref{ori_ineq_4}) result from the positivity property of the norm, and the prior knowledge of  ${\bf p}_i={\boldsymbol \beta}_i\odot \bar{\bf s}_i$. Further, instead of empirically choosing a regularization parameter $\lambda$, the objective functions can be relaxed by only using $\|\tilde{\textbf{\textit X}}\|_{2,1,2}$. The following corollary gives the guarantee for the choice of the relaxed objective function.
\begin{corollary} \label{cor1}
$\sqrt{2mL} \|\tilde{\textbf{\textit X}}\|_{2,1,2}\geq \| \tilde{\textbf{\textit X}} \|_1 \geq \| \tilde{\textbf{\textit X}}\|_*$
\end{corollary}
\begin{proof}
The nuclear norm (Schatten 1-norm) \cite{bhatia2013matrix,watrous2011theory} of matrix $\tilde{\textbf{\textit X}}$ is the sum of its singular values. By performing SVD on $\tilde{\textbf{\textit X}}= {\bf U}{\bf \Sigma} {\bf W}^H$, one can define
\begin{align}\nonumber
||\tilde{\textbf{\textit X}}||_*   \overset{\Delta}{=}  \text{inf} \{ \sum_i |  c_i | :  \tilde{\textbf{\textit X}} = \sum_i  c_i {\bf u}_i {\bf w}_i^H, \| {\bf u}_i  \| = \| {\bf w}_i  \|  =1  \},
\end{align}
where ${\bf u}_i$ and ${\bf w}_i$ are $i$-th column of unitary matrix $ {\bf U}$ and $ {\bf W}$, respectively. $c_i$ is $i$-th diagonal element of $\bf \Sigma$.
Note that $\| \tilde{\textbf{\textit X}} \|_1 = \sum_{i=1}^m \sum_{j=1}^{2LN} | \tilde{\textbf{\textit X}}_{i,j} |$, and one also can make rank one decomposition of $ \tilde{\textbf{\textit X}}$ as
\begin{align}\nonumber
  \tilde{\textbf{\textit X}}  = \sum_{i=1}^{m} \sum_{j=1}^{2LN} \tilde{\textbf{\textit X}}_{i,j} {\bf e}_i {\bf e}_j^H
\end{align}
 in terms of standard basis vectors $\{{\bf e}_i \}$. Then, we have $ \| \tilde{\textbf{\textit X}} \|_1 \geq \| \tilde{\textbf{\textit X}}\|_*$. 
We also can make $\| \tilde{\textbf{\textit X}} \|_1 = \sum_{i=1}^N \sum_{j=1}^{2mL} | \tilde{\textbf{\textit X}}_{i,j} |$, and recall that 
\begin{align}\label{2_1_2_exp} \nonumber
\|\tilde{\textbf{\textit X}}\|_{2,1,2} &\overset{\Delta}{=} \| {\bf v} \|_{1,2} =\sum_{i=1}^{N} \| [{\bf v}_{1,i}, {\bf v}_{1,i+N}, \cdots, {\bf v}_{1,i+(2L-1)N}  ] \|_2 \\  
&= \sum_{i=1}^{N} \| [ \tilde{\textbf{\textit X}}_{:,i}^T , \tilde{\textbf{\textit X}}_{:,i+N}^T , \cdots , \tilde{\textbf{\textit X}}_{:,i+(2L-1)N}^T  ] \|_2  ,
\end{align}
where we denote $\tilde{\textbf{\textit x}}_i = [ \tilde{\textbf{\textit X}}_{:,i}^T , \tilde{\textbf{\textit X}}_{:,i+N}^T , \cdots , \tilde{\textbf{\textit X}}_{:,i+(2L-1)N}^T  ] \in {\mathbb C}^{1\times 2mL}, \forall i$. Since for any vector ${\bf x} \in {\mathbb C}^n$,  $\sqrt{n} \| {\bf x} \|_2 \geq \|{\bf x} \|_1$, it follows that $\sqrt{2mL} \|\tilde{\textbf{\textit X}}\|_{2,1,2}\geq \| \tilde{\textbf{\textit X}} \|_1$, because the size of $\tilde{\textbf{\textit x}}_i $ is $2mL$.
\end{proof}

After $\tilde{\textbf{\textit X}}$ is estimated, SVD is used to obtain its eigenvector with the largest eigenvalue, which will be the best estimate of ${\bf h}$ and ${\bf x}$. However, since ${\bf x} =[\bar{\bf s}^T,{\bf p}^T]^T\in {\mathbb C^{2N\times 1}}$ is complex-valued and sparse, we only can compute the absolute value of the off-grid DoA $|\beta_i| = \frac{|p_i|}{|\bar{s}_i|}$ for non-zero $\bar{s}_i$ in terms of ${\bf p}_i={\boldsymbol \beta}_i\odot \bar{\bf s}_i$. In order to recover the sign of the off-grid deviation, all $2^K$ cases of the sign of ${\boldsymbol \beta}$ for $K$ known or detected sources must be considered. In order to determine the best estimate of the sign of off-grid DoA ${\boldsymbol \beta}$, one can calculate $ || \Phi \text{vec}(\tilde{\textbf{\textit X}})-\text{vec}({\bf Y}^T)||_2$ for all $2^K$ cases, and choose the best ${\boldsymbol \beta}$ with the minimum value. (Remember that  $\tilde{\textbf{\textit X}}_i={\bf h} {\bf x}^T_i$, ${\bf x}_i=[\bar{\bf s}_i^T,{\bf p}_i^T]^T $, and ${\bf p}_i={\boldsymbol \beta}_i\odot \bar{\bf s}_i$.) 

\subsection{Derivation of SOCP Representation}
As mentioned earlier and supported by Corollary \ref{cor1}, we can relax (\ref{2_1_2_norm}) to only minimize the objective $\|\tilde{\textbf{\textit X}}\|_{2,1,2}$. Then, this can be reformulated into second-order cone programming (SOCP)  \cite{boyd2009convex,ben2001lectures,malioutov2005} as follows:
\begin{align} \nonumber
&\arg\min_{\tilde{\textbf{\textit X}},{\bf v},{\bf z},{\bf b},q} ~~~  q \\  \label{ineq_1}
&\text{subject to  }~ || {\bf z} ||_2\leq \eta, ~~ \Phi \text{vec}(\tilde{\textbf{\textit X}})-\text{vec}({\bf Y}^T) = {\bf z}  \\ \label{ineq_2}
&~~~~~~~~~~~~~~ \| \tilde{\textbf{\textit X}}_{:,k}  \|_2 \leq {\bf v}_{1,k} , \forall k = 1, \cdots, 2LN  \\  \label{ineq_3}
&~~~~~~~~~~~~~~{\bf 1}^T{\bf b}\leq q   \\ \label{ineq_4}
&~~~~~~~~~~~~~~\sqrt{ \sum_{j=0}^{2L-1} \| \tilde{\textbf{\textit X}}_{:,i+jN} \|_2^2 }\leq b_i, \forall i=1,\cdots,N   \\ \nonumber
&~~~~~~~~~~~~~~{\bf v}\geq 0   \\ \nonumber
&~~~~~~~~~~~~~~{\bf v}_{1,l N+1:(l+1)N} \leq r{\bf v}_{1,(l-1)N+1:l N},  \\   \nonumber  
&~~~~~~~~~~~~~~~~~~~~~~~~~~~~~~~~~~~~\forall l = 1, 3, \cdots, (2L-1)     \nonumber
\end{align}
where ${\bf 1} \in {\mathbb R}^{N\times 1}$ is an all-one vector. 
\begin{proof}
The auxiliary vector $\bf z$ is used to replace the term $\Phi \text{vec}(\tilde{\textbf{\textit X}})-\text{vec}({\bf Y}^T)$ in (\ref{ori_ineq_1}) to create an equality, and a second-order cone in (\ref{ineq_1}). The auxiliary variable $q$ and vector $\bf b$ are used to replace $\|\tilde{\textbf{\textit X}}\|_{2,1,2}$.
Recall the definition of the objective function from (\ref{2_1_2_exp}) so that we have $\|\tilde{\textbf{\textit X}}\|_{2,1,2}=\sum_{i=1}^N \sqrt{ \sum_{j=0}^{2L-1} \| \tilde{\textbf{\textit X}}_{:,i+jN} \|_2^2 } $. In order to satisfy the general form of SOCP, one can use $q$ and $\bf b$ to relax it by using linear constraint (\ref{ineq_3}) and nonlinear constraint (\ref{ineq_4}), which infer to  $\|\tilde{\textbf{\textit X}}\|_{2,1,2} \leq {\bf 1}^T{\bf b}\leq q$. We also relax the feasible set of constraint (\ref{ori_ineq_2}) by using inequalities in (\ref{ineq_2}).
\end{proof}
SOCP can be solved by interior point methods, and the computational complexity ${\mathcal O}(m^{3.5}(2LN)^{3.5})$ is equal to interior point implementation cost ${\mathcal O}(m^{3}(2LN)^{3})$ per iteration times iteration complexity ${\mathcal O}(m^{0.5}(2LN)^{0.5})$.


\begin{figure*}[htb]
  \centering
  \subfigure[ES]{\includegraphics[width=\columnwidth,height=13.5em]{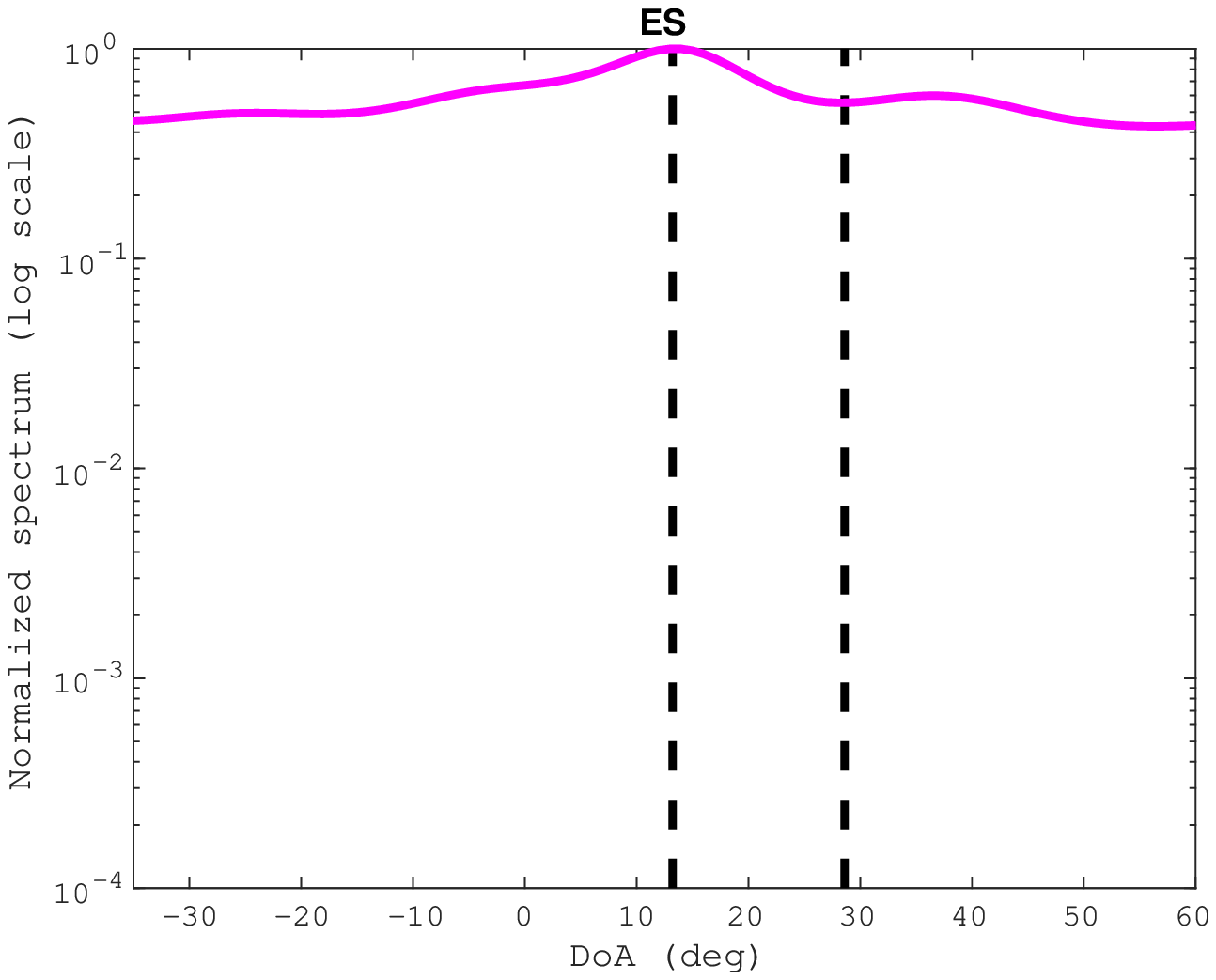}}\hfill
  \subfigure[Ling]{\includegraphics[width=\columnwidth,height=13.5em]{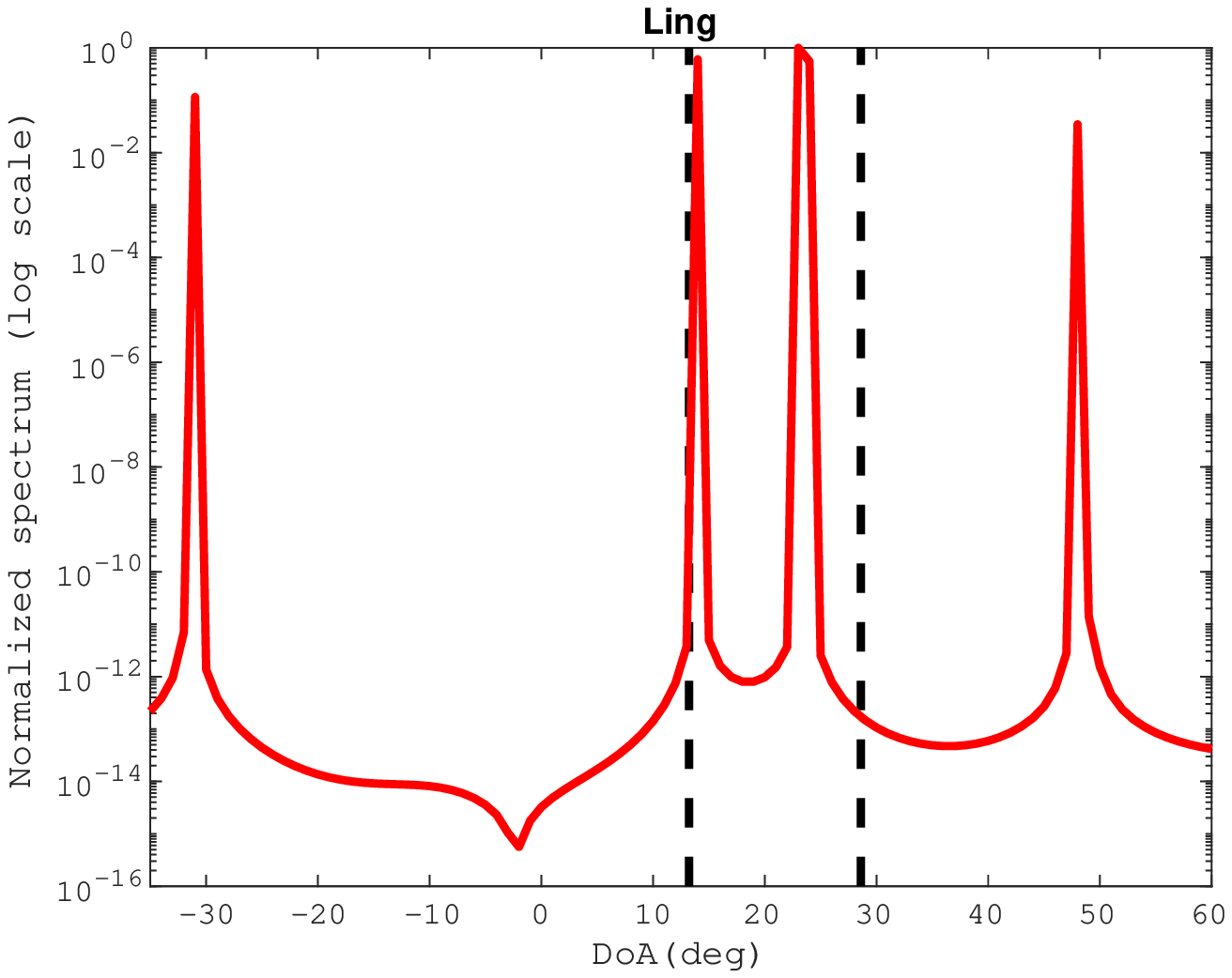}}\hfill
  \subfigure[MMV-SC]{\includegraphics[width=\columnwidth,height=13.5em]{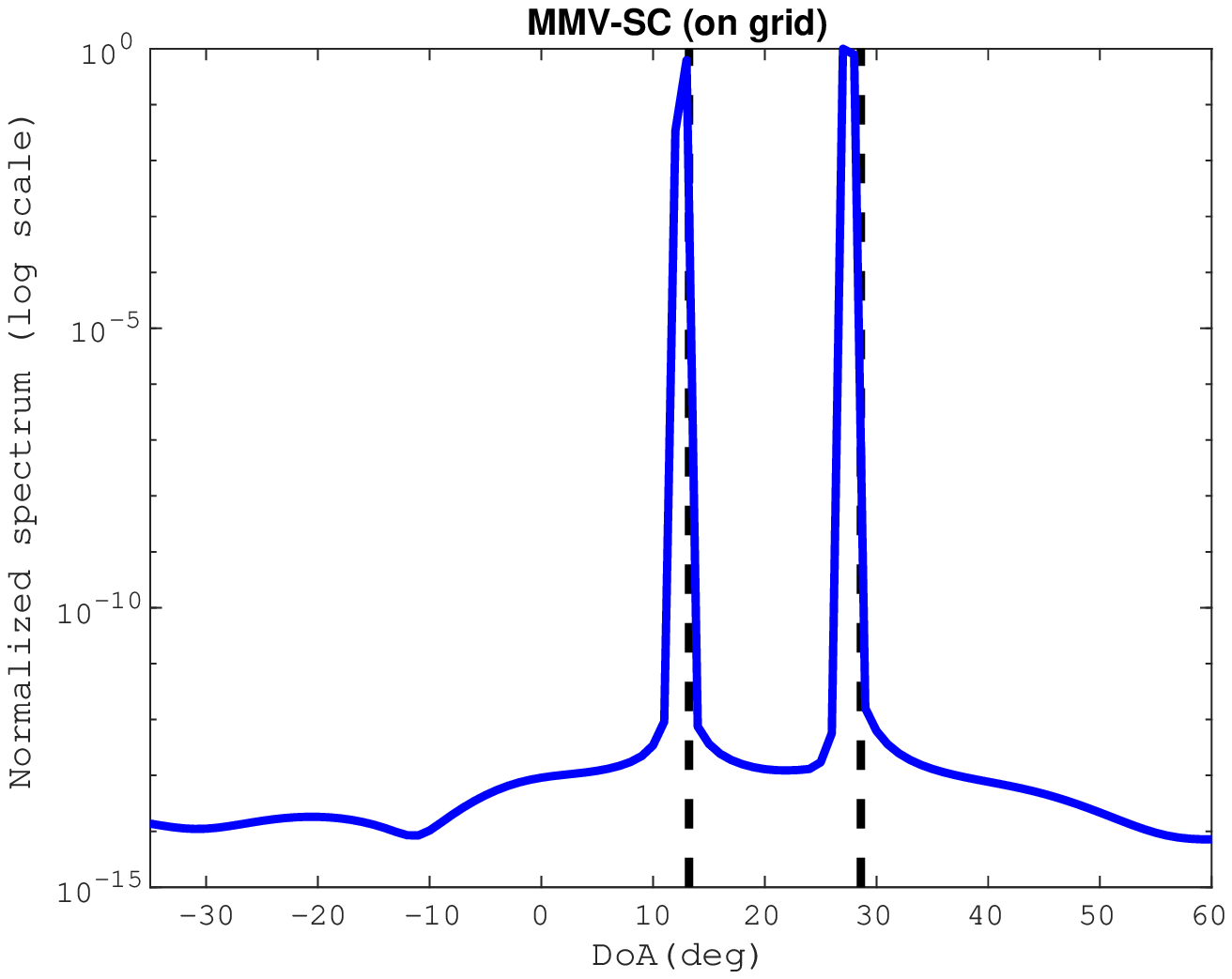}}\hfill
  \subfigure[Proposed]{\includegraphics[width=\columnwidth,height=13.5em]{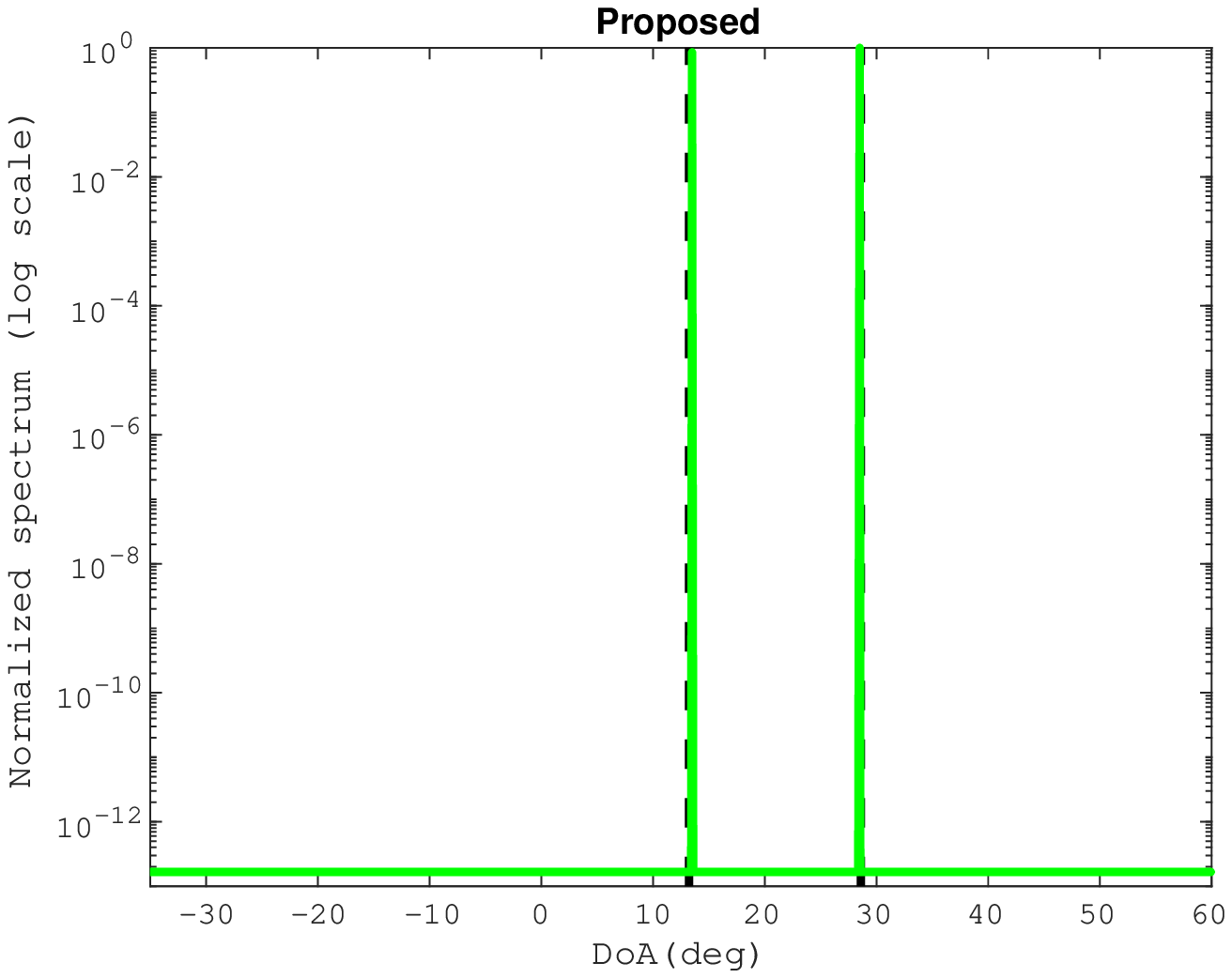}}\hfill
\caption{Performance of DoA resolution at SNR = 10 dB, $M=8$.} \label{fig:compare}
\end{figure*}

\section{Numerical Results}

In this section, numerical simulations are used to compare the performance of the proposed method with CRB \cite{friedlander2014array}, the eigenstructure (ES) method \cite{weiss1990eigenstructure},  Ling's method \cite{ling2015self}, and MMV-SC \cite{hung2017Low}.
A ULA of $M=8$ sensors with $d/\lambda=0.5$ is considered. There are $K=2$ far-field plane waves from the actual DoAs ${\boldsymbol \theta}$. We consider the off-grid case with ${\boldsymbol \theta}=[13.2220^{\circ}, 28.6022^{\circ}]$. 
 Narrowband, zero-mean, and uncorrelated sources for the plane waves are assumed, and the noise is AWGN with zero-mean and unit variance. The DoA search space is discretized from $-90^{\circ}$ to $90^{\circ}$ with $1^{\circ}$ separation, so $N=180$. The number of snapshots is set to $L=100$. The value of $r$ is set to $0.5^{\circ}$. Calibration error ${\bf d}$ is given by ${\bf d}={\bf Bh}$, where ${\bf B}\in \mathbb{C}^{M\times m}$, whose columns are the first $m=4$ columns of $M\times M$  Discrete Fourier Transform (DFT) matrix.   
We define the root mean square error (RMSE) of DoAs estimation as $(E[\frac{1}{K}\| \hat{{\boldsymbol \theta}} - {\boldsymbol \theta} \|_2^2])^{\frac{1}{2}}$. 
One hundred  realizations  are performed at each SNR.

\begin{figure}[htb]
\begin{center}
\includegraphics[width=\columnwidth]{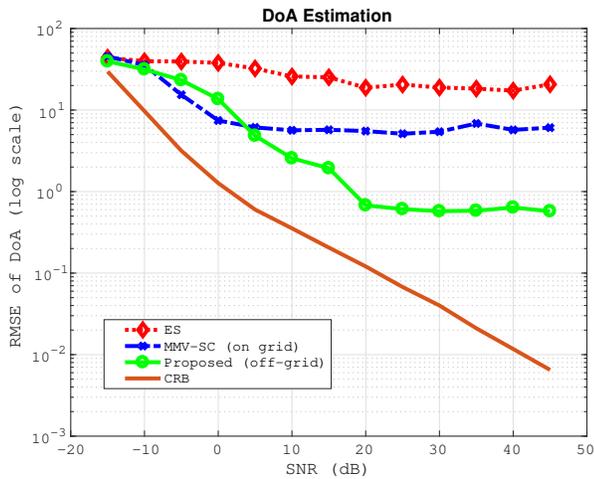}
\caption{RMSE of off-grid DoA estimation versus SNR, $M=8$.} \label{fig:ch5_RMSE_SNR_offG}
\end{center}
\end{figure}

\subsection{DoA Resolution Performance}
In the first numerical experiment, the resolution test is performed to verify the ability of estimating two closely located DoAs at SNR $=10$ dB by inspecting the normalized spectra. In Figure \ref{fig:compare}, the proposed method outperforms Ling's method due to the benefit of multiple snapshots. The proposed method also performs better than MMV-SC due to accounting for the off-grid effect. The ES method only detects one DoA of the sources.

\subsection{Off-Grid DoA Estimation Accuracy}
The relative performance of the proposed method for the off grid DoAs is demonstrated as a function of signal-to-noise ratio in Figure \ref{fig:ch5_RMSE_SNR_offG}. The proposed method outperforms the ES method at each SNR, and is also better than MMV-SC method (on grid,) which does not consider the off-grid effect. However, the RMSE performance of the proposed method saturates when SNR $\geq 20$ dB, implying that the accuracy of off-grid DoA estimation is bounded at high SNRs. The first order approximation to account for the off-grid errors is not sufficient. Hence, even the proposed method remains far away from the CRB.

\section{Conclusion}
In this work, we extended our previous work to the off-grid case, and explored the perturbation structure brought by the off-grid effect and additional information from multiple measurement vectors. We exploited the joint block-sparsity structure of $\tilde{\textbf{\textit X}}$ to  improve the accuracy performance of DoA estimation by formulating a  new convex optimization problem with nonlinear/linear inequalities constraints. The SOCP representation of the optimization problem was also derived such that efficient interior point methods can be used to obtain the numerical solution. We demonstrated the performance of the proposed method by numerical simulations.  


%
%
%
%
%
%

\ifCLASSOPTIONcaptionsoff
  \newpage
\fi



%

%
%
%

\bibliographystyle{IEEEtran}
\bibliography{IEEEabrv,mybib,mybib_ch5}

\end{document}